\documentclass{llncs}

\usepackage{
amsmath,
amscd,
amssymb,
}
\usepackage{wasysym}
\usepackage{comment}
\usepackage{tikz}
\usetikzlibrary{positioning,arrows}

\newcommand{\Z}{\mathbb{Z}}

\newcommand{\N}{\mathbb{N}}

\newcommand{\SIGMA}{\mathrm{\Sigma}}
\newcommand{\PI}{\mathrm{\Pi}}

\title{Hard Asymptotic Sets for One-Dimensional Cellular Automata\thanks{Research supported by the Academy of Finland Grant 131558}}

\author{
	Ville Salo
}
\institute{University of Turku \\ TUCS -- Turku Centre for Computer Science}

\begin{document}
\maketitle

\begin{abstract}
We prove that the (language of the) asymptotic set (and the nonwandering set) of a one-dimensional cellular automaton can be $\SIGMA^1_1$-hard. We do not go into much detail, since the constructions are relatively standard.
\end{abstract}

\section{Background}

It is well-known that (the language of) the limit set of a cellular automaton can be $\PI^0_1$-hard. Usually, \cite{Hu87} is given as the reference, since this was presumably where the result was first claimed, although the proof given is wrong. We give a similar result for the asymptotic set (hopefully with a correct proof). It turns out that asymptotic sets live in the analytical hierarchy instead of the arithmetic hierarchy, and their level is $\SIGMA^1_1$. A similar problem on countable SFTs is solved in \cite{SaTo12c}.

In \cite{DuPo09}, a two-dimensional cellular automaton with a maximally complicated asymptotic set \emph{in terms of Kolmogorov complexity} is constructed. There is no reason why the asymptotic set of this CA should have high computational complexity since these two notions are relatively orthogonal. Our best guess is that it is in $\PI^0_1$, since high Kolmogorov complexity can be checked at this level, and since it is clearly computable whether a partial Robinson tiling extends to a tiling of the plane. Similarly, our construction says nothing about the Kolmogorov complexity of the asymptotic set. We note, however, that compared to the construction in \cite{DuPo09} our construction is completely trivial.

While the proof is not interesting, the result is a bit more so, since we are not aware of many examples of `practical' $\SIGMA^1_1$-complete sets. A well-known $\PI^1_1$-complete set is the set of notations for countable ordinals, though.

\section{Asymptotic Sets}

\subsection{Asymptotic Sets on Full Shifts and SFTs with Positive Entropy}

We consider cellular automata, shift-commuting continuous functions, on the full shift $X = S^\Z$. Our reference for the analytical hierarchy is \cite{Sa90}. Our main interest is in $\SIGMA^1_1$-predicates $P(w)$ with a free variable $w$ ranging over $\N$ (usually bijected with $S^*$), since these turn out to characterize the asymptotic set. The definition of such predicates is that exactly one existential second-order (set) quantifier, no universal second-order quantifier and any number of first-order (number) quantifiers is used. For these, the natural normal form (see Part~A, Chapter~1, Theorem~1.5 in \cite{Sa90}) is
\begin{equation}\label{eq:P} P(w) = (\exists C \subset \N)(\forall m \in \N)(\exists \ell \in \N)R(C,m,\ell,w), \end{equation}
where $R$ is recursive. A predicate with subsets of $\N$ as inputs is of course said to be recursive if the corresponding Turing machine halts no matter what the set is, after inspecting some (arbitrarily long but finite) prefix of the set.

\begin{definition}
The \emph{asymptotic set} of a CA $f : X \to X$ is the set
\[ \mathcal{A}(f) = \bigcup_{x \in X} \bigcap_{n \in \N} \overline{\bigcup_{k \geq n} f^k(x)} \]
\end{definition}

This is the union of sets of limit points of $f$-orbits of configurations. Here, and in all that follows, the \emph{language} of a subset $Y$ of $S^\Z$ is the set of words that occur as $y_{[0, k-1]}$ for $y \in Y$ and $k \in \N$, even if the set if not closed (it is well-known, and easy to see, that asymptotic sets need not be closed).

\begin{lemma}
\label{lem:UpperBound}
The (language of the) asymptotic set of a CA $f : X \to X$ is always $\SIGMA^1_1$.
\end{lemma}

\begin{proof}
Given $w \in S^*$, we wish to check whether there exists $y \in \mathcal{A}(f)$ with $y_{[0, |w|-1]} = w$. This is the case if and only if there exists $x \in X$ such that $y$ with this property appears in $\bigcap_{n \in \N} \overline{\bigcup_{k \geq n} f^k(x)}$. Thus, whether $w$ is in the asymptotic set is equivalent to
\[ (\exists C \subset \N)(\forall m \in \N)(\exists \ell \in \N)R(C,m,\ell,w), \]
where $R$ checks that $\ell \geq m$, and for the configuration $y$ encoded by $C$ in some reasonable way, we have $f^\ell(y)_{[0, |w|-1]} = w$. By form, this is a $\SIGMA^1_1$ check. \qed
\end{proof}

\begin{theorem}
\label{thm:HardAsymptoticSets}
Every $\SIGMA^1_1$ subset of $\N$ can be many-one reduced to the language of the asymptotic set of some cellular automaton on a full shift. In particular, there exists an asymptotic set with a $\SIGMA^1_1$-complete language.
\end{theorem}

\begin{proof}
First, note that we can restrict to any SFT we like by adding a spreading state and having the CA introduce it when a forbidden pattern is seen. This cannot decrease the complexity of the asymptotic set. The alphabet $S$ was left unspecified in the statement of the theorem, since we can also use any (non-trivial) alphabet we like. Namely, any SFT can be recoded to one over $\{0, 1\}$ through a constant-length substitution, and we can use $0$ as the spreading state if the substitution was chosen so that $\ldots 000 \ldots$ does not appear in the image.

We use the SFT $Y$ with configurations of the form
\[ (\ldots \# \# a_0 a_1 \ldots a_j |
\begin{smallmatrix}b_0 \\ c_0\end{smallmatrix} \begin{smallmatrix}b_1 \\ c_1\end{smallmatrix}
\begin{smallmatrix}b_2 \\ c_2\end{smallmatrix} \ldots) \times Z \times A^\Z, \]
where $\#$ and $|$ are special symbols, $a_i, b_i \in \{0, 1\}$, $(c_i)_i$ is of the form $1^*00\ldots$, $A$ is a finite set of \emph{helper states} we leave unspecified, and $Z$ is composed of configurations of the form $\ldots \rightarrow \rightarrow q \leftarrow \leftarrow \ldots$, where $q \in Q$, and $Q$ is the state set of a Turing machine $M$ discussed later.

Our cellular automaton will simulate the machine $M$ on configurations of this form. The \emph{head}, marked by the $q \in Q$ on the middle track, moves around, reading values from the first track and possibly changing values on the $A^\Z$-track. On the first track, $a_i$ and $b_i$ cannot change their values, but the bits $c_i$ may be flipped from $0$ to $1$. The values $a_i$ compose the \emph{input}, the bits of $b_i$ represent the guessed set (the part $(\exists C \subset \N)$ of \eqref{eq:P}) and the guesses needed for the universal quantification (the part $(\exists \ell \in \N)$ of \eqref{eq:P}), and the values $c_i$ are used for the universal quantification itself (the part $(\forall m \in \N)$ of \eqref{eq:P}). Thus, the configuration $(b_i)_i$ encodes both an infinite subset $C$ of $\N$ and a number $\ell_i \in \N$ for each $i \in \N$. We refer to the latter numbers as a \emph{Skolemization} of the universal quantification. From now on, we leave the values of the helper states implicit, and discuss instead the projection of $Y$ where they are removed.

A \emph{signaling} configuration is a configuration of $Y$ of the form
\[ (\ldots \# \# . w | x) \times (\ldots \rightarrow \rightarrow . q_0 \leftarrow \leftarrow \ldots), \]
where $.$ denotes the origin, $q_0$ is a dedicated state of the Turing machine, $w \in \{0, 1\}^*$, and $x \in (\{0, 1\} \times \{0, 1\})^\N$. Note that the set of signaling configurations is open, since up to shifting, this is just the union of cylinders $[\# w |] \times [\rightarrow q_0 \leftarrow^{|w|}]$, where $w$ ranges over $\{0, 1\}^*$.

Given a $\SIGMA^1_1$ predicate $P(w) = (\exists C \subset \N)(\forall m \in \N)(\exists \ell \in \N)R(C,m,\ell,w)$, we choose the Turing machine and thus the cellular automaton so that $\{w \;|\; P(w)\}$ reduces to the asymptotic set via the reduction $\phi$ mapping
\[ w \mapsto (\# w |) \times (\rightarrow q_0 \leftarrow^{|w|}) \]
(the $A$-track containing, say, only unary data).

We explain what happens on a signaling configuration. First, the Turing machine exits the state $q_0$; it will not be re-entered until we explicitly state so. It then looks for the smallest $p$ such that $c_p = 0$. If one is found, the machine decodes the values $\ell_1, \ldots, \ell_p$ from $(b_i)_i$, and checks $R(C,i,\ell_i,w)$ for $1 \leq i \leq p$, decoding $C$ from $(b_i)_i$ as needed. If these checks are accepting, then $c_i$ is flipped to a $1$. The Turing machine then returns back to its original position and enters the state $q_0$ (and, say, empties the $A$-track in the progress). If something out of the ordinary happens (say, the encoding of $C$ is incorrect), a spreading state is introduced.

A configuration with $u = (\# w |) \times (\rightarrow q_0 \leftarrow^{|w|})$ at the origin is in the asymptotic set if $P(w)$ holds, since
\[ (\ldots \# \# .w | x) \times (\rightarrow . q_0 \leftarrow^{|w|}) \]
has such a configuration as a limit point if $x = \begin{smallmatrix}b_0 \\ 0\end{smallmatrix} \begin{smallmatrix}b_1 \\ 0\end{smallmatrix} \begin{smallmatrix}b_2 \\ 0\end{smallmatrix} \ldots$, if the $C$ encoded in $(b_i)_i$ is the correct quess for $w$, and the values $\ell_i$ are a corresponding Skolemization of the part $(\forall m \in \N)(\exists \ell \in \N)$ of \eqref{eq:P}.

If the spreading state ever occurs in the orbit of a point, then only words over the spreading state are added to the asymptotic set. Also, if a configuration with $u$ at the origin appears in the asymptotic set, then in particular a configuration with $u$ at the origin eventually appears. From such a configuration, the computation goes as outlined above assuming that a spreading state is not introduced. By compactness, we cannot hope for an encoding of the Skolemization $\ell_i$ where the values cannot be infinite, so that it is necessarily possible that our simulation of $M$ runs forever without finding the next $\ell_i$. In such a case, $q_0$ is of course not re-entered infinitely many times, so $u$ is not added to the asymptotic set. This means that infinitely many appearances of $u$ at the origin in fact prove that $P(w)$ holds. This concludes the proof that $\phi$ many-one reduces solutions of $P$ to the asymptotic set.

Since there exists a $\SIGMA^1_1$-hard subset of $\N$ and we can many-one reduce any $\SIGMA^1_1$ subset of $\N$ to the asymptotic set of such a cellular automaton, there exists a cellular automaton with a $\SIGMA^1_1$-hard asymptotic set. \qed
\end{proof}

We have shown that the \emph{language} of an asymptotic set can be $\SIGMA^1_1$-complete, in analogy with the limit set. Since asymptotic sets live far higher in the computability hierarchy, it seems natural to also encode configurations into subsets of $\N$ and consider the complexity of the corresponding set of subsets. We do not discuss this here.

Using Lemma~4.1 in \cite{PaSc10}, Theorem~\ref{thm:HardAsymptoticSets} seems to be extendable to any positive entropy SFT $X$. We give a rough outline of this construction: The lemma gives us, inside any positive entropy SFT (even sofic), a subshift $Y$ which is the image of a full shift in a constant-length substitution. On this subshift, we can simulate the cellular automaton constructed in Theorem~\ref{thm:HardAsymptoticSets}, using a cellular automaton $f$. Of course, there is some leftover to consider, and standard methods such as the Extension Lemma~\cite{Bo83} cannot really be used. However, as we only care about computational complexity, we can use forbidden patterns of $Y$ as spreading states.

First, in portions of the configuration containing only patterns of $Y$, $f$ is applied. Borders of such areas are moved toward the $Y$-patterns using a pigeonhold argument such as the Pumping Lemma, and by using the Marker Lemma~\cite{LiMa95} to ensure consistency of the process. In the asymptotic set of the cellular automaton $g$ obtained, we are left with only configurations where forbidden patterns of $Y$ occur with bounded gaps, and configurations over $Y$ which correspond to the $\SIGMA^1_1$-hard asymptotic set of $f$. Clearly, the asymptotic set of $g$ is then $\SIGMA^1_1$-complete, since Lemma~\ref{lem:UpperBound} naturally holds on all SFTs.

\subsection{Asymptotic Sets on Zero-Entropy SFTs}

Having dealt with the positive entropy case, it makes sense to ask what the situation is on zero-entropy SFTs. Interestingly, things are very different. Now, the natural level where asymptotic sets live is $\SIGMA^0_3$. In particular, these sets are in the arithmetic hierarchy instead of the proper analytical level $\SIGMA^1_1$. Of course, this is intuitive when one compares the normal form 
\[ P(w) = (\exists C \subset \N)(\forall m \in \N)(\exists \ell \in \N)R(C,m,\ell,w) \]
of a $\SIGMA^1_1$ predicate to the normal form
\[ P(w) = (\exists c \in \N)(\forall m \in \N)(\exists \ell \in \N)R(c,m,\ell,w) \]
of a $\SIGMA^0_3$ predicate.

\begin{lemma}
\label{lemma:AsymptoticsInRERERE}
The asymptotic set of a CA $f$ on a countable SFT $X$ is $\SIGMA^0_3$.
\end{lemma}

\begin{proof}
Given a word $w$, and again leaving encodings implicit, it is in $\SIGMA^0_3$ to check that
\[ (\exists x \in X)(\forall n)(\exists m > n)f^m(x)_{[1, |w|]} = w. \]
Namely, $X$ is countable, so a single number can encode the contents of a configuration in $X$. \qed
\end{proof}

Not all countable SFTs support a cellular automaton with a $\SIGMA^0_3$-complete asymptotic set, but some do.

\begin{theorem}
There exists a countable SFT, and a CA on it, with a $\SIGMA^0_3$-complete asymptotic set.
\end{theorem}

We refer to \cite{SaTo12c} for a proof.

It is an interesting question what the asymptotic sets of very simple SFTs look like.

\begin{question}
Is there a natural characterization of countable SFTs that support cellular automata with $\SIGMA^0_3$-complete asymptotic sets?
\end{question}

\section{Nonwandering Sets on Full Shifts}

\begin{definition}
The \emph{nonwandering set} of a CA $f : X \to X$ is the set
\[ \mathcal{N}(f) = \{x \in X \;|\; x \in \bigcap_{n \in \N} \overline{\bigcup_{k \geq n} f^k(x)} \]
\end{definition}

While $y \in X$ is in the asymptotic set of $f$ if it is a limit point of some $x \in X$, it is in the nonwandering set if it is its \emph{own} limit point. Again, the languages of such sets live in $\SIGMA^1_1$. The upper bound is proved as Lemma~\ref{lem:UpperBound}.

\begin{lemma}
\label{lem:NonwanderingUpperbound}
The language of the nonwandering set of a CA $f : X \to X$ is always $\SIGMA^1_1$.
\end{lemma}

So is the lower bound:

\begin{theorem}
\label{thm:HardNonwanderingSets}
Every $\SIGMA^1_1$ subset of $\N$ can be many-one reduced to the language of the nonwandering set of some cellular automaton on a full shift. In particular, there exists a nonwandering set with a $\SIGMA^1_1$-complete language.
\end{theorem}

\begin{proof}
The proof goes as that of Theorem~\ref{thm:HardAsymptoticSets}, except that the CA does not flip the values $c_i$ to $1$ one by one, but instead increments it as a binary counter (so that $(c_i)_i$ can now be any binary configuration in the SFT). If $P(w)$ does not hold, then $u = (\# w |) \times (\rightarrow q_0 \leftarrow^{|w|})$ is not even in the asymptotic set, seen as in the proof of Theorem~\ref{thm:HardAsymptoticSets}. If $P(w)$ does hold, then $u$ is in the nonwandering set, as the point
\[ (\ldots \# \# .w | x) \times (\rightarrow . q_0 \leftarrow^{|w|}) \]
in the proof of Theorem~\ref{thm:HardAsymptoticSets} now has \emph{itself} as a limit point. \qed
\end{proof}

\bibliographystyle{plain}
\bibliography{../../../bib/bib}{}

\def\ocirc#1{\ifmmode\setbox0=\hbox{$#1$}\dimen0=\ht0 \advance\dimen0
  by1pt\rlap{\hbox to\wd0{\hss\raise\dimen0
  \hbox{\hskip.2em$\scriptscriptstyle\circ$}\hss}}#1\else {\accent"17 #1}\fi}
\begin{thebibliography}{1}

\bibitem{Bo83}
Mike Boyle.
\newblock Lower entropy factors of sofic systems.
\newblock {\em Ergodic Theory Dynam. Systems}, 3(4):541--557, 1983.

\bibitem{DuPo09}
Bruno Durand and Victor Poupet.
\newblock Asymptotic cellular complexity.
\newblock In Volker Diekert and Dirk Nowotka, editors, {\em Developments in
  Language Theory}, volume 5583 of {\em Lecture Notes in Computer Science},
  pages 195--206. Springer, 2009.

\bibitem{Hu87}
L.~P. Hurd.
\newblock {Formal language characterizations of cellular automaton limit sets}.
\newblock {\em Complex Systems}, 1987.

\bibitem{LiMa95}
Douglas Lind and Brian Marcus.
\newblock {\em An introduction to symbolic dynamics and coding}.
\newblock Cambridge University Press, Cambridge, 1995.

\bibitem{PaSc10}
Ronnie Pavlov and Michael Schraudner.
\newblock Classification of sofic projective subdynamics of multidimensional
  shifts of finite type.
\newblock submitted.

\bibitem{Sa90}
G.E. Sacks.
\newblock {\em Higher recursion theory}.
\newblock Perspectives in mathematical logic. Springer-Verlag, 1990.

\bibitem{SaTo12c}
V.~Salo and I.~T{\"o}rm{\"a}.
\newblock Computational aspects of cellular automata on countable sofic shifts.
\newblock {\em Mathematical Foundations of Computer Science 2012}, pages
  777--788, 2012.

\end{thebibliography}

\end{document}